\documentclass[preprint,authoryear,times,3p,12pt]{elsarticle}
\journal{Renewable Energy}
\usepackage{amssymb}
\usepackage{amsmath}
\usepackage{amsthm}
\usepackage{amsfonts}
\usepackage{graphicx}
\usepackage{dsfont}
\usepackage{multirow}
\usepackage{subcaption}
\usepackage{epstopdf}

\usepackage{etoolbox}

\newtheorem{theorem}{Theorem}

\makeatletter
\patchcmd{\ps@pprintTitle}
{Preprint submitted to}
{ }
{}{}
\makeatother

\begin{document}
	
	\begin{frontmatter}
		
		\title{Identification testing via sample splitting -- an application to Structural VAR models}
		
		\author[KBO]{Katarzyna Maciejowska}
		\ead{katarzyna.maciejowska@pwr.edu.pl}

		\address[KBO]{Department of Operations Research and Business Intelligence, Faculty of Management, Wroc{\l}aw University of Science and Technology, 50-370 Wroc{\l}aw, Poland}

		\cortext[cor1]{Katarzyna Maciejowska}
		
		\date{This version: \today}

\begin{abstract}
	
	In this article, a novel identification test is proposed, which can be applied to parameteric models such as Mixture of Normal (MN) distributions, Markow Switching(MS), or Structural Autoregressive (SVAR) models. In the approach, it is assumed that model parameters are identified under the null whereas under the alternative they are not identified. Thanks to the setting, the Maximum Likelihood (ML) estimator preserves its properties under the null hypothesis. The proposed test is based on a comparison of two consistent estimators based on independent subsamples of the data set. A Wald type statistic is proposed which has a typical $\chi^2$ distribution. Finally, the method is adjusted to test if the heteroscedasticity assumption is sufficient to identify parameters of SVAR model. Its properties are evaluated with a Monte Carlo experiment, which allows non Gaussian distribution of errors and mis-specified VAR order. They indicate that the test has an asymptotically correct size. Moreover, outcomes show that the power of the test makes it suitable for empirical applications.
 
\end{abstract}

\begin{keyword}
	identification test\sep  Maximum Likelihood\sep  Wald test\sep  Structural VAR\sep  heterosecedasticity
\end{keyword}

\end{frontmatter}	
\section{Introduction}\label{Sec:Introduction}

Identifications is one of the crucial assumptions, which needs to be fulfilled in order to use properly the Maximum Likelihood estimation method (see \cite{Hsiao:1983}, \cite{Rothenberg1971}, among others). For majority of econometric models, this property is easily proved and hence is not a concern for researchers. Unfortunately, there are families of models, in which the identifications is not guarantee, for example Structural Vector Autoregressive (SVAR) models (\cite{Sims:1980}, \cite{Lutkepohl05}, \cite{Kilian:Lutkepohl:2017}), nonlinear regime switching models such as Smooth Transition Regression (STR) or mixture models (\cite{MCLachlan:etal:2018}, \cite{Dijk:Terasvirts:2002}). 

In \cite{Rothenberg1971}, two types of identification types are discussed: global and local. Both definitions are based on a parametric likelihood function of a random variable, $y\in R^M$. Lets denote by $l(\theta, y)$ a log-likelihood functions computed for the parameter vector $\theta$. Then
it is said that the the true parameter vector $\theta_0$ is \textit{globally identifiable} if for 
\begin{equation}
\forall_{\theta\in \Theta} \forall_{y\in R^K}: \theta_0=\theta \Leftrightarrow l(\theta_0;y)=l(\theta;y),
\end{equation}
 Global identifications implies that the ML estimator is unique and hence under regularity conditions, the Maximum Likelihood (ML) estimator of $\theta$ is consistent. Although desirable, the global identification is not necessary to estimate the model parameters. The sufficient condition is local identification.
The parameter vector $\theta_0$ is \textit{locally identifiable} if 
\begin{equation}
\forall_{\varepsilon}\forall_{\|\theta_0- \theta\|<\varepsilon} \forall_{y\in R^K}: \theta_0=\theta \Leftrightarrow l(\theta_0,y)=l(\theta,y).
\end{equation}
When the local identification is ensured, the parameter space, $\Theta$, can be restricted in such way that the parameters becomes unique. 
One of the examples of models which are identifiable only localy is a Mixture of Normal distributions (MN) model. It can be shown that when the ordering of components in the mixture is changed, which  is called a 'label switching', the parameter vectors change as well but the value of the likelihood functions remain unaffected. This problem can be solved by assuming that the probability of the first state is larger than 0.5 or that the mean of the first component is larger than the second one.

Unfortunately, there are cases, in which the local identification is not satisfied. For example, when the MN model is fitted to data stemming from a simple normal distribution or when the SVAR model is applied without imposing sufficient restrictions. In such case, there are infinitely many parameter vectors, which  provide the same value of the log-likelihood function. As the result, it is not guaranteed that the ML estimator converge to $\theta_0$ as the sample size increases.  Moreover, as shown by \cite{Rothenberg1971}, the information matrix for $\theta_0$ is singular and hence the typical ML test (Wald, LM or LR) either cannot be computed or have an unknown limiting distribution.

In the literature, there are a few identification tests dedicated to particular econometric problems. For example, in case of switching regression such as Markov Switching (MS), STR or structural break modes, it is known as a linearity test because for linear processes the models are not identified. \cite{Luukkonen:etal:1988} use a Taylor expansion of the STR model to verify if the underlying process is linear. \cite{Andrews:1993} explores \emph{Sup} transformation of classical ML statistics and proves its asymptotic properties. \cite{Hansen:1992a} proposed a test for MS models, which is based on LR statistic but has a nonstandard distribution. For SVAR models, the problem is directly called identification and has been studied by \cite{Lutkepohl:Wozniak:2020} and \cite{Lutkepohl:etal:2021}. However, there is no single testing approach, which will be suitable to more general family of parametric models and would be straightforward to apply. \cite{Andrews:Ploberger:1994} and \cite{Hansen:1996} develop tests for cases when the nuisance parameters are presented only under the alternative hypothesis. Unfortunately, their statistics have nonstandard asymptotic distributions, which cannot be tabulated accept special cases. Hence, applying these tests requires conducting additional, complex simulations. Moreover and cannot be applied to some models, for example SVAR.

In this research, a novel approach is proposed for testing  identification, which can be applied to different parametric models. It differs from the existing methods, for example a recent test of \cite{Lutkepohl:etal:2021} or previous \cite{Hansen:1996} and \cite{Andrews:Ploberger:1994} tests, in terms of the hypothesis setting. In the presented approach, similar to \cite{Rossi:Inoue:2011}, it is assumed that the model is identified under the null and not-identified under the alternative. Hence the test supports identification hypothesis, as long as there are no strong evidence against it. This is convenient for researchers who want to use for MN or SVAR models when the data does not reject them.
 Additionally, it can be notice that under the null, there should be not estimation problems and the information matrix resulting from ML method is non-singular.

The testing procedure explores the fact that under the null any consistent estimator should converge to its true value, which is uniquely identified. In order to obtain two independent estimators, the sample is split randomly into two disjoint sets. Next the two estimators are compared with the Wald type of test. This approach has a few advantages. First, it can be applied to any parametric model if sufficient data is available. Second, since the Wald type of test is based on the information matrix rather than its inverse, the statistic could be computed even when the Hessian is singular (under the alternative). Finally, the test has a well known $\chi^2$ distribution and hence is does not require any additional simulations and is straightforward to apply.

This article demonstrates how the approach can be adopted for testing identifications in SVAR models. Recent publications (see \cite{LanneLutkepohl2008}, \cite{LanneMaciejowskaLutkepohl2010} and \cite{Lutkepohl:Netsunajev:2017}, among others) show that statistical data properties, such as heteroscedasticity, can be explored to identify structural parameters without imposing additional restriction. The crucial assumption in this procedure is that shock variances change in different proportions. Although this assumption is very important, it cannot be directly tested. Some papers still use the standard Wald or Likelihood Ratio (LR) tests to evaluate the model, for example \cite{LanneMaciejowskaLutkepohl2010}, \cite{Lutkepohl:Velinov:2016}. Unfortunately, a  $\chi^2$ distribution does not provide a proper approximation of their asymptotic behavior. In other research (\cite{Lutkepohl:Wozniak:2020}), Bayesian methods are used to verify the model assumptions. Recently, \cite{Lutkepohl:etal:2021} proposed a direct frequentist test for assessing the change of volatility in SVAR models. This article  explores the properties of eigen values of covariance matrices to construct a test with an asymptotic $\chi^2$ distribution. Albeit its novelty and straightforward application, the new test is designed only for models with discrete shift of variances and it is not examined how it works for Markov-switching (MS) or STR types of heteroscedasticity. Moreover, as under the null the model is not identified, it is not clear how to state the hypothesis when the evaluated process has a dimension larger than two. The authors propose to apply sequential testing. A method proposed in this article does not suffer from the above problems. There is one null hypothesis, under which the structure of SVAR models is well identified. The test statistic has an asymptotic $\chi^2$ distribution and can be directly applied to different types of heteroscedasticity.

Finally, the performance of a newly proposed identification test is evaluated with a MC experiment. In order to provide results comparable with the previous outcomes of \cite{Lutkepohl:etal:2021}, the same data generating processes are considered. They cover the cases of non-Gaussian distributions, a misspecified order of VAR model and processes of a dimension larger than two. The results indicate that the proposed test has a proper size. Its power increases as the p-values form different sample splitting are combined via a harmonic mean. Moreover, it is robust to model misspecification and works well for higher dimension processes. Finally, the approach is applied to verify the identification of \cite{Kilian:2009}. The results confirm findings of \cite{Lutkepohl:Netsunajev:2014} and indicate that there are no strong evidences that the model is not well identified via heteroscedasticity.

The article is structured as follows. In the first Section \ref{Sec:Test} the proposed identification test is described. Next, its application to SVAR model is described in Section \ref{Sec:SVAR}. The performance of the test is evaluated by a Monte Carlo experiment, which results are shown in Section \ref{Sec:MC}. Section \ref{Sec:Empirical example} presents an exmaple of an empirical application of the testing approach and Section \ref{Sec:Conclusions} concludes.

\section{Identification test}\label{Sec:Test}

Lets consider a parametric model, which is described by a log-likelihood function $L(\theta,Y)=\sum_{t=1}^Tl(\theta, y_t)$, where $Y = \{y_1,...,y_T\}$ is the information set and $T$ is the  sample size. The parameter vector $\theta\in R^M$. In the test, two hypothesis are considered: under the null the model parameters are assumed to be identifiable whereas under the alternative they lack identification. As the result, under $H_0$ (if the likelihood functions satisfied standard regularity conditions (\cite{Greene:18})) the ML estimator becomes consistent and asymptotically normal. These properties are not maintained for $H_1$, under which the estimator does not converge in probability to $\theta_0$. The proposed approach explores the consistency property to derive the test statistic.

 In order to build the test statistic, the sample is split into two disjoint parts and the estimators coming from these two sub-samples are compared with the Wald type of test. The testing procedure consists of the following steps
\begin{enumerate}
	\item Randomly split the information set into two disjoint parts of equal sizes: $Y_1$, $Y_2$, such that $Y_1\cup Y_2=Y$, $Y_1\cap Y_2=\varnothing$
	\item Compute ML estimator for the whole information set and the two parts: $\hat{\theta}$, $\hat{\theta}_1$, $\hat{\theta}_2$, respectively
	\item Compute the average Hessian of the full sample log-likelihood functions for $\hat{\theta}$
	\begin{equation}
	H_T(\hat{\theta}) = \frac{1}{T}\sum_{t=1}^T{\frac{\partial^2l(\theta; y_t)}{\partial \theta \partial \theta'}(\hat{\theta})}
	\end{equation}
	\item Compute the Wald type statistic
	\begin{equation}\label{eq:W}
	W = T/2(\hat{\theta}_1-\hat{\theta}_2)'(-H_T(\hat{\theta})/2)(\hat{\theta}_1-\hat{\theta}_2)
	\end{equation}
\end{enumerate}
It can be shown that under the null the statistic $W$ converges to the $\chi^2(M)$ distribution, where $M$ is the size of the parameter vector $\theta$.

\begin{theorem}
	Suppose a model described by a log-likelihood function $l(\theta, y_t)$ satisfies regularity conditions, which ensure consistency of the ML estimator and its asymptotic normality. Then under the null hypothesis of parameter identification, the test statistic 
		\begin{equation*}
	W = T/2(\hat{\theta}_1-\hat{\theta}_2)'(-H_T(\hat{\theta})/2)(\hat{\theta}_1-\hat{\theta}_2)
	\end{equation*}
	 has an asymptotic $\chi^2(M)$ distribution.
\end{theorem}
\begin{proof}
	First, it can be observed that $\hat{\theta}_1$ and $\hat{\theta}_2$ are independent because they come from two different, randomly drown sub-samples. Second, as the model is identifiable under the null then the estimators are consistent and the CLT implies that

	\begin{equation}
	\sqrt{T/2}(\hat{\theta}_i-\theta_0) \rightarrow_d N(0,(-H(\theta_0))^{-1}),
	\end{equation}
	where $i=1,2$ and $H(\theta_0) = p\lim_{T\to \infty}H_T(\theta_0)$. The difference of two independent normal random variables is also normal, with the variance being the sum of the individual variances. Hence
	\begin{equation}
	\sqrt{T/2}(\hat{\theta}_1-\theta_1) \rightarrow_d N(0,2(-H(\theta_0))^{-1}),
	\end{equation}
	Finally, as $\hat{\theta}$ is a consistent ML estimator with $\hat{\theta}  \rightarrow_p \theta_0$ then also $H_T(\hat{\theta}) \rightarrow_p H_T(\theta_0)$ and finally
	\begin{equation}
	H_T(\hat{\theta}) \rightarrow_p H(\theta_0).
	\end{equation}
	As a result, 
	\begin{equation}
T/2(\hat{\theta}_1-\hat{\theta}_2)'(-H_T(\hat{\theta})/2)(\hat{\theta}_1-\hat{\theta}_2) \rightarrow_d \chi^2(M).
	\end{equation}
\end{proof}

The statistic $W$ has an intuitive interpretation. It shows that if the model is identifiable, the estimators stemming from two datasets should be similar and hence their difference converges to zero. It can be noticed that if the parameter estimates are not identifiable, the two estimators $\hat{\theta}_1$ and $\hat{\theta}_2$ do not necessary converge to the same value and hence their difference preserves not negligible even for large samples.

\subsubsection{Behavior of test under the alternative}
Behavior of the test under the alternative is a complex and nontrivial problem. When ML estimator is not consistent then the limit of $\hat{\theta}$ does not exist, and so the limits of $\hat{\theta}_1$ and $\hat{\theta}_2$. Then  the difference $\hat{\theta}_1-\hat{\theta}_2$ does not converge in probability to a zero vector.

\begin{theorem}
	Suppose there exists a set $A \subset R^M$, which includes at least two elements, such that if and only if $\tilde{\theta}\in A$, the log-likelihood $l(\tilde{\theta};y)=l(\theta_0;y)$. So the parameters in $A$ are observable equivalent to the true parameter vector $\theta_0$. Then the difference of two independent ML estimators, $\hat{\theta}_1-\hat{\theta}_2$, does not converge in probability to zero. 
\end{theorem}
\begin{proof}
	First, lets denote by $e$, a distance of a parameters vector $\theta$ to a set $A$ defined as 
	$$ e = \theta - arg\min_{\tilde{\theta}\in A}{|\theta - \tilde{\theta}|}$$
Then the parameter can be expressed as a sum 
$\theta = \tilde{\theta}+e$, with $\tilde{\theta} \in A$. Since $\theta_0\in A$ then the ML estimator will converge to $A$, meaning that $p\lim \hat{e} = 0$.

Lets now consider the difference $\hat{\theta}_1-\hat{\theta}_2$. It could be expressed as 
$$ \hat{\theta}_1-\hat{\theta}_2 = \tilde{\theta}_1-\tilde{\theta}_2 + \hat{e}_1-\hat{e}_2.$$
Since $p\lim \hat{e} = 0$ then also $p\lim (\hat{e}_1-\hat{e}_2)=0$. So the difference $ \hat{\theta}_1-\hat{\theta}_2$ will converge in probability to zero if and only if $p\lim (\tilde{\theta}_1-\tilde{\theta}_2) = 0$. The parameters $\tilde{\theta}_1$ and $\tilde{\theta}_2$ are equally likely, hence
$$ prob(|\tilde{\theta}_1-\tilde{\theta}_2|>\varepsilon) =\frac{||\tilde{\theta}_1-\tilde{\theta}_2|>\varepsilon|}{|A|}.$$
It can be noticed that since $A$ contains at least two elements and $\tilde{\theta}_1$ and $\tilde{\theta}_2$ are independent then for $\varepsilon = 0$ the probability $prob(\tilde{\theta}_1=\tilde{\theta}_2) =0$
and hence
$$ prob(|\tilde{\theta}_1-\tilde{\theta}_2|>0) \geq 0.5.$$
Hence there exists $\varepsilon$ such that 
$$ prob(|\tilde{\theta}_1-\tilde{\theta}_2|>\varepsilon) \nrightarrow 0 .$$
This implies that the difference $\hat{\theta}_1-\hat{\theta}_2$ does not converge in probability to zero.
\end{proof}

Although the distance between the two estimators $\hat{\theta}_1$ and $\hat{\theta}_2$ does not converge to zero, it does not  directly imply that the $W$ statistic under the alternative grows big enough to exceed the critical values. This is due to the fact that 
the lack of identification is associated with a singularity of a second derivative matrix, $H$ \cite[][]{Rothenberg1971}. Therefore, the behavior of the test statistic should be examined separately for different types of models and identification problems.

\subsection{Multiple splitting}


In the proposed approach, the data split is random -- not driven by the data itself. Hence it seems natural to repeat the splitting for w few times. Lets denote by $N$ the number of times the data is divided and the statistic $W_n$ with a corresponding p-value, $v_n$ is calculated. Here $n=1,2,...,N$.

Following \cite{Vovk:Wang:2020}, the p-values will be combined via averaging. As stated in the paper, the harmonic mean is a robust approach, which yields good results for both correlated and uncorrelated inputs.  
The main concern associated with application of the theory is the fact that the $p$-values are obtained with the asymptotic rather then actual distribution of the test statistic. Hence, some of the assumptions stated by  \cite{Vovk:Wang:2020} are not met. In particular, for small samples the $p$-values are not uniformly distributed. Therefore, the harmonic mean is applied here to the truncated sample $v_n\in V$, which includes only the observations of $v_n$ staying between the 20th and 80th percentiles of its empirical distribution. The truncation removes from the average elements, which are either very big (when $v_n$ is close to zero), or close to one. This change may have a positive impact on both the size and the power of the test, as it reduces symmetrically the input noise.

As the result, the averaged $p$-value, $\bar{v}_N$, is calculated according to the following formula
\begin{equation}\label{Eq:v}
\bar{v}_N = a_{N}\left(\frac{1}{\tilde{N}}\sum_{v_n\in V}{\frac{1}{v_n}}\right)^{-1}.
\end{equation}
It can be noticed that once the set $V$ is changed, the inverse of the $p$-values is averaged over a different number of observations, $\tilde{N}$. Here, $\tilde{N}=0.6N$. Finally, as stated by \cite{Vovk:Wang:2020}, the harmonic mean should be scaled by a factor $a_{N}$, which is
\begin{equation}
a_{N}=\frac{(z^*+N)^2}{(z^*+1)N,}
\end{equation}
where $z^*$ is the unique solution ($z\in(0,\infty)$) to the equation
\begin{equation*}
z^2=N((z+1)ln(z+1)-z).
\end{equation*}
The scaling factor is shown to $a_{N} \leq e \ln N$ and $a_{N}/\ln{N} \to 1$. Unfortunately, the convergence is very slow hence $a_{N}$ is approximated numerically separately for each $N$. The table with values of $a_{N}$ for selected number of averaged $p$-values can be found in \cite{Vovk:Wang:2020}.

It is expected that the theoretical properties of the truncated harmonic mean can be different from the untruncated one. However simulation results indicate its validity and importance for small samples, in which the distribution of $p$-values deviates significantly from the uniform one. The robustness of the averaging method to the selection of $V$ is discussed and evaluated in Section \ref{Sec:Robustness}.

\section{Test for an identification of a SVAR model}\label{Sec:SVAR}

\subsection{SVAR model: identification and estimation}

The identification problem is one of the central issues in SVAR modeling. The Vector Autoregressive (VAR) model describes the join behavior of a set of time series. Lets define a vector of endogenous variables as $Y_t\in R^K$. The reduced form of SVAR model takes the following form
\begin{equation}\label{eq:VAR}
Y_t = \mu + \sum_{p=1}^P A_pY_{t-p} + \varepsilon_t,
\end{equation}
where $\mu$ is a $K\times 1$ vector of intercepts, $A_p$ are $K\times K$ matrices describing the autoregresive structure and $\varepsilon_t$ is a $K\times 1$ vector of residuals. Since $\varepsilon_t$ are forecast errors, they are allowed to be correlated and hence their variance, $\Sigma$, is not diagonal.

The reduced form model, although useful for modeling of the expected value of endogenous variables, is not suited for risk analysis -- for example impulse responses or simulations. Therefor, a structural extension of the model (\ref{eq:VAR}) is used.
The SVAR model may be described in various ways (see \cite{Lutkepohl05} for an extension discussion). Here, the B-model is presented, in which the residuals,  $\varepsilon_{t}$, are a linear function of structural shocks, $u_{t}$. The structural shocks are assumed to be uncorrelated and have a diagonal variance-covariance matrix, $\Lambda$. Then
\begin{equation}
\label{eq:Structural_errors}
\varepsilon_{t}=Bu_{t}.
\end{equation}
The matrix $h$ is called \emph{an instantaneous effect matrix}, because it describes how structural shocks affect endogenous variables in the current time period, $t$. For example, $B_{ij}$ describes the impact of the $j$-th structural shock, $u_{j,t}$ on the $i$-th element of $Y_{t}$.
Notice that equation (\ref{eq:Structural_errors}) implies that
\begin{equation}\label{eq:Variance}
\Sigma=B\Lambda B',
\end{equation}
 so there is a direct relationship between $B$ and the variance of errors $\varepsilon_{t}$. It is typically assumed that either structural shocks have an identity variance-covariance matrix, $\Lambda=I$, or the diagonal elements of $B$ are equal to one. Here, the first approach is adopted.

Unfortunately, the parameters of SVAR model cannot be directly estimated because there are infinitely many matrices $B$, which satisfies (\ref{eq:Variance}).  As shown in the literature \cite[see][for a comprehensive discussion on VAR models]{Lutkepohl05}, the structural model requires estimation of $K^2$ elements of the $B$ matrix. At the same time, the variance-covariance matrix $\Sigma$ of the reduced form consists of only $K(K+1)/2$ distinct parameters. Therefore additional $K(K-1)/2$ restrictions need to be imposed, which will restrict the parameter space. 

The identification restrictions may stem either from the theoretical knowledge of the analyzed economic problem or from statistical properties of the data. \cite{LanneLutkepohl2008} shows that if there is a shift in the variance of structural shocks, one can recover the $B$ matrix without imposing additional constraints. \cite{LanneLutkepohl2008} assumes that the shift is due to a structural change, which appears at time $T_c$, and hence for $t\leq T_c$ the variance of shocks $\Sigma_u=I$, whereas for $t>T_c$ there is $\Sigma_u=\Lambda$. They stated that if the diagonal elements of $\Lambda$ are all different then the model parameters are identifiable. Unfortunately, this assumption could not be tested, because if some of the elements were identical, the model would become not identifiable and standard ML based statistics could not be used.

Using the results from previous Section \ref{Sec:Test}, a two step procedure, similar to \cite{Lutkepohl:etal:2021}, is proposed to estimate SVAR model. In the first step, the parameters of the reduced form of the model (\ref{eq:VAR}) are estimated with the Generalized Least Square (GLS) method. As the result, the residuals $\hat{\varepsilon}_t=Y_t-\hat{\mu}-\sum_{p=1}^P\hat{A}_pY_{t-p}$ are computed. In the second step, the structure model of errors,  $\hat{\varepsilon}_t$,  is estimated with the ML method and its identification is tested. In ML, it is assumed that shocks are independent and Gaussian. As the result, the residuals of the reduce form VAR follows:
 \begin{equation} \label{eq:e}
\varepsilon_t \sim \left\{ \begin{array}{lcl}
N(0,BB') & \mbox{for} & t\leq T_c\\ 
N(0,B\Lambda B') & \mbox{for} & t>T_c. \\
\end{array}\right.
\end{equation}
The log-likelihood function used in ML becomes
\begin{equation}
l(\theta) = -T\log\det(B)-\frac{T-T_c}{2}\log \det(\Lambda)-\frac{1}{2}\sum_{t=1}^{T_c}\varepsilon_t'(BB')^{-1}\varepsilon_t-\frac{1}{2}\sum_{t=T_c+1}^{T}\varepsilon_t'(B\Lambda B')^{-1}\varepsilon_t,
\end{equation}
where $\theta = (vec(B)', \lambda_1,...,\lambda_K)'$ is a vector of parameters ('vec' is an operator, which stacks columns of the given matrix into a single vector). As shown by \cite{Lutkepohl:etal:2021}, if we denote $\Sigma_1=BB'$ and $\Sigma_2=B\Lambda B'$ then diagonal elements of $\Lambda$ are the eigenvalues of the matrix $\Sigma_2\Sigma_1^{-1}$. The columns of the matrix $B=[b_1,...,b_K]$ are proportional to the eigenvectors satisfying the equation:
\begin{equation}
\left(\Sigma_1^{-1} - \lambda_k\Sigma^{-1}_2\right)b_k=0,
\end{equation}
for $k=1,2,..,K$. As shown by \cite{LanneMaciejowskaLutkepohl2010}, when the eigenvalues are all distinct then the matrix $B$ is unique up to a permutation and a sign of its columns. Therefore, in order to ensure that the two estimators: $\hat{\theta}_1$ and $\hat{\theta}_2$ corresponds to each other, it is assumed that $\lambda$'s are in a descending order ($\lambda_1\geq\lambda_2\geq ... \geq \lambda_K$) and that an element of each column $b_k$ with the largest absolute value is positive.

\subsection{Identification test}
In order apply the testing procedure described in Section, the following pair of hypothesis is considered
\begin{equation*}
H_0: \forall k\in\{1,...,K-1\}: \lambda_k \neq \lambda_{k+1}
\end{equation*}
\begin{equation*}
H_1: \exists k\in\{1,...,K-1\}: \lambda_k=\lambda_{k+1}
\end{equation*}
It can be noticed that the hypothesis are opposite to those stated by \cite{Lutkepohl:etal:2021}, which assumes that the model is identifiable under the alternative. As the result, \cite{Lutkepohl:etal:2021} needed to consider a set of different null hypothesis for models with $K>2$ and test the identifications sequentially. Here, the null hypothesis is unique and hence it is sufficient to compute the test only once.

The testing procedure consists of the following steps
\begin{enumerate}
	\item  Split randomly the sample $S=\{1,2,...,T\}$ into $S_1$, $S_2$ in such way that: $|S_1|=|S_2|=T/2$ and $|S_1\leq T_c| = |S_2\leq T_c|=T_c/2$, where $|S_i\leq T_c|$ is the number of observations such that $S_i\leq T_c$.
	\item Estimate the model parameters for the  whole sample, $\hat{\theta}$, and the two sub-samples $\hat{\theta}_1$, $\hat{\theta}_2$. 
	\item Compute the average second derivative of the Log Likelihood function for $\hat{\theta}$: $H_T(\hat{\theta})$
	\item Compute the test statistic
	\begin{equation}
	W=-T/4(\hat{\theta}_1-\hat{\theta}_2)'H_T(\hat{\theta})(\hat{\theta}_1-\hat{\theta}_2)
	\end{equation}	
	and compare it with $\chi^2(K^2+K)$. When the null is not rejected, there is no evidence that the model is not identifiable. 
\end{enumerate}	
When the underlying distribution of structural shocks is non Gaussian then the likelihood function is misspecified. In particular, $H_T(\hat{\theta})$ may not converge to its true values. For heavy tails distributions, such as $t$-Student distribution, the parameter variances can be underestimated and the $W$ statistic may reject the null too often. Since the Hessian of $l(\theta)$ depends on the inverse of the square of error variance, whereas the covariance of parameters depend on the fourth central moment of errors, the discrepancy between $E\varepsilon_t^4$ and $(E\varepsilon_t^2)^2$ is partially responsible for the misperformance of the test. Therefor, a kurtosis -- $\kappa$ -- of the underlying DGP seems crucial for adjusting the test to the non-Gaussian distributions. Here, a modified version of the $W$ statistic is proposed, in which 
\begin{equation}
\tilde{H}_T(\theta, \kappa) = H_T(\theta)/(\kappa/3).
\end{equation} 
and
	\begin{equation}
W(\hat{\kappa})=-T/4(\hat{\theta}_1-\hat{\theta}_2)'\tilde{H}_T(\hat{\theta}, \hat{\kappa}) (\hat{\theta}_1-\hat{\theta}_2)=-T/4(\hat{\theta}_1-\hat{\theta}_2)'H_T(\hat{\theta})(\hat{\theta}_1-\hat{\theta}_2)/(\hat{\kappa}/3),
\end{equation}
where $\hat{\kappa}$ is a consistent estimator of $\kappa$. In this article, an estimator discussed by \cite{Schott:2001}  and applied in \cite{Lutkepohl:etal:2021} is used in order to make the results comparable.
It could be noticed that for normally distributed residuals, $\kappa=3$ and $W=W(3)$. Moreover, for any consistent estimator of kurtosis, such that $\hat{\kappa} \to_p \kappa$, $W(\hat{\kappa})\to_d \chi^2(K^2+K)$.

Finally, the steps(1)-(4) can be repeated for $N$ times, as discussed in Section \ref{Sec:Test}, and the average $p$-value, $\bar{v}_N$ can be computed according to (\ref{Eq:v}). In may happen that  $\bar{v}_N>1$. In such case, the $\bar{v}_N$ is set to equal to one.

\section{Small sample properties of the identification test}\label{Sec:MC}

In order to assess the finite sample properties of the proposed test, a Monte Carlo (MC) experiment is conducted. It is assumed that structural shocks, $u_t$ are i.i.d. white noise processes, with a covariance matrix characterized by a structural change in period $T_c=\tau T$. Similar to \cite{Lutkepohl:etal:2021},  bivariate and three-dimensional GDPs are considered:
 
 \begin{itemize}
 	\item $\mathbf{DGP1}$: In the model $y_t=\varepsilon_t$ and $\varepsilon_t=Bu_t$. The structural shocks $u_t\in R^2$ are independent Gaussian or $t(5)$ distributions with the covariance matrix $\Lambda=diag(\lambda_1, \lambda_2)$ and the instantaneous effect matrix is $B=I_2$. The structural change occurs at $\tau=0.5$. Analogous to \cite{Lutkepohl:etal:2021}, two values of $\Lambda$ are considered: $(\lambda_1, \lambda_2) = (2,2), (2,1)$. 
 	
 	\item $\mathbf{DGP2}$: A two-dimensional VAR(2) model
 		$$ y_t = \left[ \begin{array}{c}
 		0.190 \\
 		0.523\\
 		\end{array}\right]+
 		\left[ \begin{array}{cc}
 	0.-0.036 &-0.705\\
 	-0.093& 1.211\\
 	\end{array}\right]y_{t-1}+
 	\left[ \begin{array}{cc}
 	0.090 &0.796\\
 	-0.085& -0.276\\
 	\end{array}\right]y_{t-2}+\varepsilon_t.$$
 	The residuals $\varepsilon_t=Bu_t$. The structural shocks $u_t\in R^2$ are independent Gaussian with $\Lambda=diag(\lambda_1, \lambda_2)$ and $(\lambda_1, \lambda_2) = (0.5,0.5), (0.5,0.1)$. The structural change occurs at $\tau = 0.3$ and the instantaneous effect matrix is
 	$$ B = \left[ \begin{array}{cc}
 	0.317 &1.059\\
 	0.242& -0.450\\
 	\end{array}\right].$$
 	
 	\item $\mathbf{GDP3}$: A three dimensional VAR(0) model with $y_t=\varepsilon_t$ and $\varepsilon_t=Bu_t$. The structural shocks $u_t\in R^3$ are independent Gaussian with $\Lambda=diag(\lambda_1, \lambda_2, \lambda_3)$ and $(\lambda_1, \lambda_2, \lambda_3) = (3,2,1), (3,2,2), (2,2,2)$. The structural change occurs at $\tau = 0.5$ and the instantaneous effect matrix $B = I_3$.
 	
 	 \item $\mathbf{GDP4}$:  A three dimensional VAR(0) model with $y_t=\varepsilon_t$ and $\varepsilon_t=Bu_t$. The structural shocks $u_t\in R^3$ are independent Gaussian with $\Lambda=diag(\lambda_1, \lambda_2, \lambda_3)$ and $(\lambda_1, \lambda_2, \lambda_3) = (2,0.4,0.2), (2,0.2,0.2), (1,1,1)$. The structural change occurs at $\tau = 0.4$ and the instantaneous effect matrix $B$ is.
 	$$ B = \left[ \begin{array}{ccc}
 	27.92    &0.231    &1.569\\
 	0.441    &5.643   & 0.079\\
 	0.496    &0.643   &-4.668\\
 	\end{array}\right].$$
 	The choice of DGP4 is inspired by the empirical example used in the Section \ref{Sec:Empirical example}.
 	\end{itemize}
 	
In order to enable a direct comparison of the MC results with the recent outcomes presented by \cite{Lutkepohl:etal:2021}, the same DGPs are selected for analysis. This simulations are extended by DGP 4, which helps to evaluate findings obtained during empirical application of the method.

In this novel testing procedure, the structural model is fitted to two sub-samples of the size $T/2$, so the sample sizes used in MC are doubled as compared to \cite{Lutkepohl:etal:2021}. As the result, $T\in \{200, 500, 1000\}$ is considered. Finally, to examine the asymptotic properties of the test, the results for $T=2000$ are also analyzed. The number of Monte Carlo replications is set to equal 1000. The sample splitting is repeated for 100 times, when the multiple split approach is used.
 
\subsection{Non-Gaussian residuals}
First, let us analyze the results of DGP1 for two different methods of selecting the kurtosis value: $\kappa=3$ and $\kappa=\hat{\kappa}$. The fist case corresponds to the assumption of residuals normality, whereas the second one represents the case, in which the residuals are potentially non-normal.
Table \ref{tab:MC:GDP1} presents the outcomes for a single split (W) and a combined multiple split (AveW) statistics. Columns 3-6 and 7-10 show the empirical size and the empirical power of the test, respectively, for the nominal significance level 5\%.

First, let us analyze the results for Gaussian shocks. It can be observed that for short samples ($T=200$) the basic, $W$, test has a size exceeding 5\%. As the sample size increases, the empirical size converges to its nominal level. Hence, the test is asymptotically valid. At the same time, the size of the $AveW$ test is lower than for $W$ and falls to zero once the sample reaches 1000 observations. This reflects the fact that harmonic mean provides conservative p-values (see \cite{Vovk:Wang:2020}). When the power of these tests are considered, it can be noticed that the $W$ test has relatively poor power (reaching 50\% for $T=1000$). At the same time, the power $AveW$ is much higher and exceeds 80\% for $T=2000$.

The results are quite different, when the test is used for a non Gaussian distribution. If the errors are $t$-distributed then their kurtosis differ from $3$ and $W$ statistic is misspecified. As the result, $W$ is considerably oversized even in large samples. For example, for $T=2000$, the test rejects the null in more than 10\% cases for 5\% significance level. In contrast, when $\kappa=\hat{\kappa}$ is allowed the rejection frequencies becomes much closer to the significance level when the null hypothesis is true. For example, for $T=1000$ the test reject $H_0$ in 7\% cases and one could observe a slow convergence to the nominal level. The power of the test is comparable with that for Gaussian processes.

Finally, one can not observe large differences between Gaussian and t-Student distributions for $AveW$ test. In case of non Gaussian errors, its size is slightly larger than previously and exceeds 5\%, but the difference diminish once a longer sample ($T=500$) is used. Moreover, the power of the $AveW$ remains at a similar level even when the distribution changes. As the result, the $AveW$ outperforms basic $W$ test in terms of the power level and the robustness to the error distribution.

In the remaining part of the article, it is assumed that the kurtosis parameter is estimated, $\kappa=\hat{\kappa}$.

\begin{table}
	\caption{Empirical size and power of identification tests $W$ and $AveW$ , DGP1}
	\label{tab:MC:GDP1}
	\centering
	
\begin{tabular}{c|c|cc|cc|cc|cc}
 \multicolumn{2}{c}{}&	\multicolumn{4}{|c}{Size} &	\multicolumn{4}{|c}{Power}\\ 
\multicolumn{2}{c}{} &	\multicolumn{4}{|c}{ $(\lambda_1,\lambda_2)=(2,1)$} &	\multicolumn{4}{|c}{ $(\lambda_1,\lambda_2)=(2,1)$}\\ 
\cline{1-10}
Kurtosis	&Sample	&\multicolumn{2}{|c}{Gaussian} 	& \multicolumn{2}{|c}{$t(5)$} &\multicolumn{2}{|c}{Gaussian} 	& \multicolumn{2}{|c}{$t(5)$}\\
\cline{3-10}
 ($\kappa$)	&($T$)	 	&	W	&	AveW	&	W	&	AveW &	W	&	AveW	&	W	&	AveW	\\

\hline
3	&	200	&	0.126	&	0.069	&	0.180	&	0.114&	0.366	&	0.532	&	0.348	&	0.503	\\
	&	500	&	0.077	&	0.004	&	0.132	&	0.019&	0.459	&	0.644	&	0.404	&	0.639	\\
	&	1000	&	0.045	&	0.000	&	0.100	&	0.001	&	0.487	&	0.745	&	0.490	&	0.723	\\
	&	2000	&	0.056	&	0.000	&	0.111	&	0.000	&	0.552	&	0.829	&	0.505	&	0.786	\\
	\hline
$\hat{\kappa}$	&	200	&	0.112	&	0.068	&	0.133	&	0.094	&	0.365	&	0.510	&	0.382	&	0.534	\\
	&	500	&	0.059	&	0.006	&	0.084	&	0.016	&	0.427	&	0.637	&	0.483	&	0.676	\\
	&	1000	&	0.055	&	0.000	&	0.072	&	0.000&	0.494	&	0.740	&	0.532	&	0.760	\\
	&	2000	&	0.054	&	0.000	&	0.060	&	0.000	&	0.559	&	0.833	&	0.565	&	0.852\\
	\hline
	
	\end{tabular} 	

\vspace{0.2cm}
\small{Note: for $AveW$: number of splits is $N=100$ and the truncation is at 20\% and 80\% percentiles}	
\end{table}

\subsection{Estimation of VAR parameters}
 
Next, let us analyze the impact of the estimation of VAR parameters on the test performance. DGP2 is generated as a VAR(2) process, which parameters correspond to the estimates of the \cite{BlanchardQuah1986} model fitted to updated time series, as in \cite{Chen:Natsunajev:2016}. As noticed by the \cite{Lutkepohl:etal:2021}, the performance of the identification test may depend on the proper selection of the autoregression order. Therefore, the testing procedure was conducted for both VAR(2) and VAR(1) specifications. The results are presented in Table \ref{tab:MC:GDP2}. 

The outcomes of both $W$ and $AveW$ tests are similar to DGP1. One cannot observe any significant change of the size or the power of these test when the misspecified VAR(1) model is fitted. In case of both VAR(2) and VAR(1) models, the size of the basic $W$ test is close to the nominal level, whereas the $AveW$ test is too conservative. The power of the $W$ oscillates around 50\% for long samples. Finally, the power of $AveW$ test reaches 80\% for $T=2000$.

\begin{table}
	\caption{Empirical size and power of identification tests $W$ and $AveW$, DGP2}
	\label{tab:MC:GDP2}
	\centering
	
	\begin{tabular}{c|cc|cc|cc|cc}
	
\multicolumn{1}{c}{} &	\multicolumn{4}{|c}{Size} &	\multicolumn{4}{|c}{Power}\\ 
\multicolumn{1}{c}{}  &	\multicolumn{4}{|c}{ $(\lambda_1,\lambda_2)=(0.5,0.1)$} &	\multicolumn{4}{|c}{ $(\lambda_1,\lambda_2)=(0.5,0.5)$}\\ 
		\cline{1-9}
	Sample		&\multicolumn{2}{|c}{VAR(1)} 	& \multicolumn{2}{|c}{VAR(2)} &\multicolumn{2}{|c}{VAR(1)} 	& \multicolumn{2}{|c}{VAR(2)}\\
		\cline{2-9}
		($T$)	 	&	W	&	AveW	&	W	&	AveW &	W	&	AveW	&	W	&	AveW	\\
		
		\hline
	200	&	0.063	&	0.000	&	0.067	&	0.000&	0.366	&	0.637	&	0.340	&	0.588	\\
	500	&	0.060	&	0.000	&	0.054	&	0.000&	0.429	&	0.709	&	0.428	&	0.712	\\
	1000	&	0.036	&	0.000	&	0.052	&	0.000&	0.494	&	0.788	&	0.484	&	0.758	\\
	2000	&	0.052	&	0.000	&	0.055	&	0.000&	0.552	&	0.836	&	0.513	&	0.832	\\
		\hline

	\end{tabular} 	

\vspace{0.2cm}
\small{Note: kurtosis parameter is $\kappa=\hat{\kappa}$; for $AveW$: number of splits is $N=100$ and the truncation is at 20\% and 80\% percentiles}	
\end{table}

\subsection{Processes of higher dimension}

In many empirical problems, the researchers need to analyze higher-dimension processes. The approach proposed by \cite{Lutkepohl:etal:2021} experiences some problems in case of three-dimensional Gaussian DGP due to a need of sequential testing when a type of the identification violation is not known. The method presented in this article solves some of the previous issues as it is asymptotic valid regardless of the process dimension. Clearly, it is expected that long samples are needed to estimate the model parameters accurately but the asymptotic properties should be maintained under the null hypothesis. 

The simulation results for the three dimensional case are presented in Table \ref{tab:MC:GDP3}-\ref{tab:MC:GDP4}. In these tables, the results of $W$ and $AveW$ approaches are accompanied by the outcomes of the $Q$ test proposed by \cite{Lutkepohl:etal:2021}. It should be notices that in case of $Q$, there are three potential null hypothesis ($\lambda_1 = \lambda_2$, $\lambda_2 = \lambda_3$, $\lambda_1 = \lambda_2= \lambda_3$). The  $p$-value presented in the is tables is the smallest $p$-value of the these three tests. Moreover, it should be rememberd that the hypothesis of $Q$ test are stated differently: under the null, the model is not identified whereas under the alternative it is. Hence, the presented rejections of null have different interpretation: for $W$ and $AveW$ tests they show how often the approach indicates lack of identification  and for $Q$ it shows the opposite.

For GDP3 the parameters $(\lambda_1,\lambda_2, \lambda_3) = (3,2,1)$ represents the process under the null, whereas $(\lambda_1,\lambda_2, \lambda_3) = (2,2,2)$ or $(3,2,2)$ indicate the model is not identifiable. For the DGP4, the null is described by $(\lambda_1,\lambda_2, \lambda_3) = (2,0.4,0.2)$. The results confirm that the test is oversize for short samples. As expected, the empirical size converges to the nominal level as the sample grows but substantially exceeds 5\% for $T<2000$. When the two test  $W$ and $AveW$ are considered, the results supports previous findings. The multiple split approach is conservative as its empirical size converges to zero. Moreover, it has substantially higher power than the $W$ test, particularly for $(\lambda_1,\lambda_2, \lambda_3)$ equal to $(3,2,2)$ for GDP3 and $(2,0.2,0.2)$ for GDP4. 

Next, let us compare the presented approaches with the $Q$ statistic. It can be noticed first that the \cite{Lutkepohl:etal:2021} test is too conservative as its empirical size falls short the nominal 5\%. Second, for very short samples, $T=200$ it may have almost no power: it indicates lack of identification in 95.8\% of cases when $(\lambda_1,\lambda_2, \lambda_3) = (3,2,1)$. Finally, as sample grows, the power of the test rises substantially and reaches almost 100\% for the longest samples $T=2000$.

The results indicate that for short samples, both testing approaches are prone to indicate falsely lack of identification. Hence, it should be used with some concern. Despite this fact, the $AveW$ test can be useful when evaluating the null even for $T=200$, as its size adjusted power is and 0.474 and 0.428 for DGP3 and DGP4, respectively.

\begin{table}
	\caption{Empirical size and power of tests $W$, $AveW$ and $Q$ for DGP3}
	\label{tab:MC:GDP3}
	\centering
	\begin{tabular}{c|c|cc|c}
	Parameters&	T &W & AveW & Q\\
		\hline
(3,2,1)	&	200	&	0.319	&	0.394	&	0.042	\\
&	500	&	0.180	&	0.143	&	0.470	\\
&	1000	&	0.113	&	0.038	&	0.824	\\
&	2000	&	0.072	&	0.02	&	0.992	\\
\hline
(3,2,2)	&	200	&	0.543	&	0.868	&	0.000	\\
&	500	&	0.548	&	0.852	&	0.000	\\
&	1000	&	0.537	&	0.821	&	0.000	\\
&	2000	&	0.567	&	0.829	&	0.000	\\
\hline
(2,2,2)	&	200	&	0.650	&	0.957	&	0.000	\\
&	500	&	0.765	&	0.989	&	0.003	\\
&	1000	&	0.843	&	0.994	&	0.013	\\
&	2000	&	0.881	&	1.000	&	0.045	\\

		\hline

	\end{tabular}

\vspace{0.2cm}
\small{Note: kurtosis parameter is $\kappa=\hat{\kappa}$; for $AveW$: number of splits is $N=100$ and the truncation is at 20\% and 80\% percentiles; for $Q$ test under $H_0$ the model parameters are not identified}
\end{table} 

\begin{table}
	\caption{Empirical size and power of tests $W$, $AveW$ and $Q$ for DGP4}
	\label{tab:MC:GDP4}
	\centering
	\begin{tabular}{c|c|cc|c}
		Parameters&	T &W & AveW & Q\\
		\hline
		(2,0.4,0.2)	&	200	&	0.148	&	0.087	&	0.560	\\
		&	500	&	0.071	&	0.003	&	0.938	\\
		&	1000	&	0.045	&	0.000	&	1.000	\\
		&	2000	&	0.054	&	0.000	&	1.000	\\
		
		\hline
		(2,0.2, 0.2)	&	200	&	0.369	&	0.515	&	0.048	\\
		&	500	&	0.422	&	0.623	&	0.050	\\
		&	1000	&	0.472	&	0.713	&	0.059	\\
	&	2000	&	0.536	&	0.813	&	0.039	\\
	
		\hline
		(1,1,1)	&	200	&	0.774	&	0.993	&	0.000	\\
		&	500	&	0.828	&	1.000	&	0.003	\\
		&	1000	&	0.846	&	0.998	&	0.013	\\
		&	2000	&	0.881	&	1.000	&	0.000	\\
		
		\hline
		
	\end{tabular}

\vspace{0.2cm}
\small{Note: kurtosis parameter is $\kappa=\hat{\kappa}$; for $AveW$: number of splits is $N=100$ and the truncation is at 20\% and 80\% percentiles; for $Q$ test under $H_0$ the model parameters are not identified}
\end{table} 

\subsection{Robustness analysis of $AveW$ approach}\label{Sec:Robustness}

In the $AveW$ test, the $p$-values, $v_n$, stemming from different sample splits are averaged via the harmonic mean. Before the mean is computed, the extreme values of $v_n$ are removed. In the above results, 20\% lowest and highest observations are truncated. This rises a question, how the adopted truncation threshold impacts the results. Therefore, different thresholds: (10\%,90\%) and (30\%, 70\%) are applied to the DGP1. The results are presented in Table \ref{tab:Appendix}, in which the first column shows the rejection frequency for the specification previously studied. Next columns describe the outcomes of different specifications. They indicate that the size and the power depends on the truncation used. For  (10\%,90\%) the null is rejected more often than for (20\%,80\%) and (30\%,70\%). This leads to a growth of power at the cost of exceeding the nominal size. When the adjusted power is analyzed, it shows that widening of the truncation bounds may be preferable. Although the results are quantitatively different, they reflect a similar pattern. Generally, the $AveW$ test is conservative,  its power increases with the sample size and is larger than for $W$ test.

Another important decision, which needs to be taken when applying $AveW$ is the choice of $N$: the number of $v_n$ to be averaged. In the above MC, $N$ is set to 100. In order to evaluate the robustness of the outcomes to the selection of $N$, two other values are used: $N=50, 200$. The resulting rejections of the null hypothesis for DGP1 are presented in the last two columns of Table \ref{tab:Appendix}. They show that the outcomes are robust. Change of $N$ and does not alter significantly neither empirical size nor power. 

To sum up, when applying the $AveW$ method, the researcher can choose any $N$, as long as its big enough. It should however select the truncation thresholds more carefully, taking into account that it may increase the size and result in an over liberal approach, as for $T=200$ and (10\%,90\%). 

\begin{table}
	\caption{Rejection frequencies of the null hypothesis of $AveW$ for different truncation thresholds and numbers the averaged $p$-values, $N$}
	\label{tab:Appendix}
	\centering
	\begin{tabular}{c|c|c|cc|cc}
		
		& Truncation& (0.2,0.8)& (0.1,0.9)& (0.3,0.7)& (0.2,0.8)& (0.2,0.8)\\
		\cline{2-7}
		$(\lambda_1,\lambda_2)$	&T/N & 100& 100& 100& 50& 200\\
		\hline
		(2,1)	&	200	&	0.068	&	0.146	&	0.034	&	0.06	&	0.069	\\
		(size)&	500	&	0.006	&	0.018	&	0.000	&	0.004	&	0.002	\\
		&	1000	&	0.000	&	0.001	&	0.000	&	0.000	&	0.000	\\
		\hline
		(2,2)	&	200	&	0.532	&	0.694	&	0.377	&	0.497	&	0.522	\\
		
		(power)&	500	&	0.644	&	0.814	&	0.516	&	0.662	&	0.666	\\
		&	1000	&	0.745	&	0.877	&	0.597	&	0.732	&	0.747	\\
		\hline
		(2,2)	&	200	&	0.464	&	0.548	&	0.343	&	0.437	&	0.453	\\
		(adjusted power)&	500	&	0.638	&	0.796	&	0.516	&	0.658	&	0.664	\\
		&	1000	&	0.745	&	0.876	&	0.597	&	0.732	&	0.747	\\
		\hline
	\end{tabular}

\vspace{0.2cm}
\small{Note: kurtosis parameter is $\kappa=\hat{\kappa}$}	
	
\end{table}

\section{Empirical example}\label{Sec:Empirical example}
The empirical example corresponds to the \cite{Kilian:2009} model of a global oil market. In \cite{Kilian:2009} a VAR model is applied to describe the joint behavior of three variables: $y_t=(\Delta prod_t, q_t, p_t)$, where $\Delta prod_t$ is the percentage growth of the global oil production, $q_t$ is the logarithm of a detreded index of real economic activity and $p_t$ is the logarithm of the real oil price. In the article, three structural shocks are analyzed: oil supply shock, aggregated demand shock and an oil-market-specific demand shock. In order to identify the shocks, the instantaneous effect matrix, $B$, was assumed to be lower triangular.

The model has been evaluated by \cite{Lutkepohl:Netsunajev:2014}, who explore heteroscedasticity to identify the structural parameters. In the paper, the monthly data from 1973m2--2006m12 are used to estimate a Markow Switching model. The results indicate that there are only weak evidence of a lower triangular form of the matrix $B$. This conclusion was next challenged by \cite{Lutkepohl:etal:2021}, where the shift of shocks volatility in October 1987 was used to identify the model. The outcomes suggest that the Kilian model can be only partially identified via the variance change and hence the lower-triangularity test should be compared with $\chi^2(2)$ distribution, rather than $\chi^2(3)$. These results show that identification testing can be crucial for interpreting the model outcomes.

Here we use the data explored by \cite{Lutkepohl:etal:2021} and apply the proposed testing procedure to verify if the change of shocks volatility is sufficient to indentify the model structure. Similar to previous papers, a VAR(3) model is fitted to the data with the GLS method. The estimated relative variances are $\lambda_1=0.186$,  $\lambda_2=0.363$ and $\lambda_3=2.157$. The full samples and sub-samples estimator of the randomly split sample are presented in Table \ref{tab:EX:param}. Although the estimators $\hat{\theta}_1$ and $\hat{\theta}_2$ are distinct, they are very similar, which supports the null hypothesis.

\begin{table}
	\caption{Estimates of model parameters for the whole sample ($\hat{\theta}$) and a randomly split sub-samples ($\hat{\theta}_1$, $\hat{\theta}_2$).}
	\label{tab:EX:param}
	\centering
	\begin{tabular}{c|ccc}
	
	Parameter	& $\hat{\theta}$& $\hat{\theta}_1$& $\hat{\theta}_2$\\
		\hline
$B_{11}$&27.918&   30.493 &  24.879\\
$B_{12}$&0.441  &  0.233  &  0.539\\
$B_{13}$&0.496  &  0.267  &  0.694\\
$B_{21}$&0.231  & -0.825  &  1.756\\
$B_{22}$&5.643  &  5.258  &  6.003\\
$B_{23}$&0.643  &  0.089  &  1.137\\
$B_{31 }$&1.569  &  1.203  &  1.792\\
$B_{32 }$&0.079  &  0.028  &  0.122\\
$B_{33 }$&-4.668 &  -3.378 &  -5.624\\
\hline
$\lambda_1$&0.186  &  0.184  &  0.189\\
$\lambda_2$&0.363  &  0.404  &  0.332\\
$\lambda_3$&2.157  &  4.173  &  1.461\\
\hline
	\end{tabular} 		

\end{table} 

Finally, the $W_n$ statistics are computed and compared with $\chi^2(12)$ distribution. The individual and averaged p-values are $v_1=0.5096$ and $\bar{v}=1$, which indicate that there are no strong evidences against the null hypothesis. The results support initial conclusions of  \cite{Lutkepohl:Netsunajev:2014} and question the lower-traingularity of the $B$ matrix. However, it should be noticed here that the Monte Carlo results show that the size adjusted power of the test for $T=500$ is 0.620, which is much below 1. Therefore, the conclusions presented  are not very strong and should be repeated once more data becomes available.

\section{Conclusions}\label{Sec:Conclusions}

In this article, a novel test is proposed to verify if the parametric model is identifiable. Unlike in previous papers, the test null hypothesis assumes model identification, whereas under the alternative the model parameters are not identified. This modifications allows to use the well known results of the ML estimation method, which states that under some regularity conditions ML estimators are consistent and asymptotically normal. In order to explore these features, two parameter estimators are compared, based on a random split of the sample. Under the null, both estimators are consistent and hence the distance between them converges to zero. The significance of the differences is assessed with a Wald type of test, which has a well known $\chi^2$ distributions with the number of degrees of freedom corresponding to the number of parameters. Next a modification of the approach is proposed, which explores outcomes from multiple splitting of the sample. The resulting $p$-values are combined with a truncated, scaled harmonic mean.

Next, the method is adjusted for testing whether heteroscedasticity is sufficient to identify parameters of a SVAR model. 
In order to assess the  properties of the test, a MC experiment is conducted. The data generating processes used here correspond to those applied by \cite{Lutkepohl:etal:2021}. The results of the analysis indicate that for short samples, the $W$ test is slightly over-sized. However, as the sample size increases, the empirical size converges to its nominal level. The power of the test depends on the dimension of the VAR model. It improves as the number of identical variances increases. It is also demonstrated that the multiple split test, $AveW$, is conservative, with its empirical size falling short of the nominal level. At the same time, the power to $AveW$ exceeds the power of $W$ counterpart. Moreover, it is more robust to VAR model misspecification and non-Gaussian distribution of residuals. These results shows that $AveW$ is a reasonable choice in empirical applications. 

I believe that the theory can be developed in various directions. First, it can be applied to SVAR models with more complex heteroscedasticity structure, as in \cite{Maciejowska:etal:2016}, \cite{Lutkepohl:Netsunajev:2014}. In such case a different sampling method should be considered. Moreover, it can be used to assess many other models, such as MN, MS or STR, which are affected by the identification problem. It would be desirable to evaluate the  empirical properties of the test in the context of a broad family of applications.


\bibliographystyle{elsarticle-harv}
\bibliography{bibliography}

\end{document}